\documentclass[12pt,a4paper]{article}
\usepackage{url}
\usepackage{cancel}
\usepackage{color}
\usepackage[normalem]{ulem}
\usepackage[english]{babel}
\usepackage{color}
\usepackage[english]{babel}
\usepackage{color}
\usepackage{graphicx}
\usepackage{amsmath,amssymb,amsthm}
\usepackage[T1]{fontenc}
\usepackage[utf8]{inputenc}
\usepackage{authblk}
\usepackage[english]{babel}
\usepackage{amsmath,amssymb,amsthm}
\usepackage{color}
\usepackage{graphics}
\usepackage{amsmath,amssymb,amsthm}
\newtheorem{definition}{Definition}
\newtheorem{theorem}[definition]{Theorem}
\newtheorem{lemma}[definition]{Lemma}
\newtheorem{corollary}[definition]{Corollary}
\newtheorem{remark}[definition]{Remark}
\newtheorem{proposition}[definition]{Proposition}
\newtheorem{example}{Example}

\renewcommand{\tilde}[1]{\widetilde{#1}}

\newcommand{\ff}{\mathbb{F}_q}

\title{Relative generalized Hamming weights of one-point
  algebraic geometric codes\thanks{%
The result in this paper is in part submitted for
possible presentation in IEEE Information Theory Workshop
(ITW 2014) \cite{itw2014}.}}
\author[1]{Olav Geil\thanks{olav@math.aau.dk}}
\author[1]{Stefano Martin\thanks{stefano@math.aau.dk}}
\author[2]{Ryutaroh Matsumoto\thanks{ryutaroh@rmatsumoto.org}}
\author[1]{Diego Ruano\thanks{diego@math.aau.dk}}
\author[3]{Yuan Luo\thanks{yuanluo@sjtu.edu.cn}}
\affil[1]{Department of Mathematical Sciences, Aalborg University, Denmark}
\affil[2]{Department of Communications and Computer Engineering,
    Tokyo Institute of Technology, Japan}
\affil[3]{Computer Science and Engineering Department, Shanghai Jiao
  Tong University, China}

\begin{document}
\maketitle

\begin{abstract}
Security of
linear ramp secret sharing schemes can be characterized by the relative
generalized Hamming weights of the involved
codes~\cite{luo,kurihara}. In this paper we elaborate on the
implication of these parameters and we devise a method to estimate
their value for general one-point algebraic geometric codes. As it is
demonstrated, for Hermitian codes our bound is often
tight. Furthermore, for these codes the relative generalized Hamming
weights are often much larger than the corresponding generalized
Hamming weights.\\

\noindent {\bf{Keywords:}} linear code, Feng-Rao bound, Hermitian code, one-point algebraic geometric code, relative dimension/length profile, relative generalized Hamming weight, secret sharing, wiretap channel of type II.
\end{abstract}

\section{Introduction}
A secret sharing scheme is a cryptographic method to encode a secret
into multiple shares later distributed to participants, so that only
specified sets of participants can reconstruct the secret. The first
secret sharing scheme was proposed by Shamir~\cite{shamir1979share}. It was a
perfect scheme, in which a set of participants unable to reconstruct
the secret has absolutely no information on the secret. Later,
non-perfect secret sharing schemes were proposed~\cite{Blakley,Yamamoto} in which
there are sets of participants that have non-zero amount of
information about the secret but cannot reconstruct it. The term ramp
secret sharing scheme is sometimes used for the latter mentioned
type of schemes, sometimes for the union of the two types. In this
paper we will apply the most general definition, but concentrate our
investigation on non-perfect secret sharing schemes. Secret sharing
has been used, for example, to store confidential information to
multiple locations geographically apart. By using  secret sharing
schemes in such a scenario, the likelihoods of both data loss and data
theft are decreased. As far as we know, in many applications both
perfect and non-perfect ramp secret sharing schemes can be used. In
the perfect scheme, the size of a share must be at least that of the
secret~\cite{capocelli1993size}. On the other hand, ramp secret sharing schemes allow
shares to be smaller than the secret, which is what we concentrate on
in this paper. Such schemes are particularly useful for storing bulk
data~\cite{csirmaz2009ramp}\\ 

A linear ramp secret sharing scheme
can be described
as a coset
construction $C_1/C_2$ where $C_2 \subsetneq C_1$ are linear codes \cite{Chen}
. It
was shown in~\cite{Bains,kurihara,subramanian2009mds} that the corresponding relative
dimension/length profile (RDLP)
expresses the worst case information
leakage to unauthorized sets in such a system.
RDLP was proposed by Luo et~al.\ \cite{luo}.
They \cite{luo} also proposed the relative generalized
Hamming weight (RGHW) and its equivalence to RDLP, 
similar to the one demonstrated by Forney \cite{forney94}
between the dimension/length profile and the generalized
Hamming weight.
The $m$-th RGHW expresses the smallest size of
unauthorized sets that can obtain $m$ $q$-bits \cite{Bains,kurihara},
where $q$ is the size of the alphabet of $C_2 \subsetneq C_1$.
In order to investigate the potential of
linear codes to construct useful ramp secret sharing schemes,
it is indispensable to study the RGHW and the RDLP.
However, not much research has been done so far,
partly because the connection between the secret sharing
and RGHW/RDLP was only recently reported.
In particular, few classes of linear codes have been examined for
their RGHW/RDLP.
In this paper we study RGHW of general linear codes
by the Feng-Rao approach \cite{agismm}, and explore
its consequences for one-point algebraic geometry (AG) codes \cite{tsfasmanvladut,handbook} and
in particular the Hermitian codes \cite{tiersma,stichtenothhermitian,yanghermitian}.\\

The present paper starts with
a discussion of known results regarding linear ramp secret sharing
schemes and it continues with demonstrating that the RGHWs can also be
used to express the best case information leakage. The main result of
the paper is a method to estimate RGHW of one-point algebraic
geometric codes.
This is done by carefully applying the Feng-Rao
bounds
\cite{agismm}
for primary \cite{AG}
as well as dual \cite{FR1,FR2,pellikaan93efficient,handbook,MM,heijnenpellikaan}
codes. From this we derive a
relatively simple bound which uses information on the corresponding
Weierstrass semigroup \cite{hurwitz,centina}. As shall be demonstrated for Hermitian
codes the new bound is often sharp. Moreover, for the same codes the RGHW are often much larger than the corresponding
generalized Hamming weights (GHW) \cite{wei}
which means that studies of RGHW cannot be substituted by
those of GHW.

The paper is organized as follows. Section~\ref{sec2} describes the
use of RGHW in connection with linear ramp secret sharing schemes, and
in connection with communication over the wiretap channel of type
II. In Section~\ref{sec3} we apply the theory to the special case of
MDS codes. In Section~\ref{sec4} we show -- at the level of general
linear codes -- how to employ the Feng-Rao bounds to estimate RGHW. This
method is then applied to one-point algebraic geometric codes in
Section~\ref{sec5}. We investigate Hermitian codes in
Section~\ref{sec6}, and treat the corresponding ramp secret sharing
schemes in Section~\ref{sec7}.

\section{Ramp secret sharing schemes and wiretap channels of type II}\label{sec2}
Ramp secret sharing schemes were introduced in
\cite{Blakley,Yamamoto}. Let $\ff$ be the finite field with $q$
elements. A ramp secret sharing scheme	with $t$-privacy and
$r$-reconstruction is an algorithm that, given an input $\vec{s} \in
\ff^\ell$, outputs a vector $\vec{x} \in \ff^n$, the vector of shares
that we want to share among $n$ players, such that, given a collection
of shares $\{x_i \mid i \in \mathcal{I}\}$ where ${\mathcal{I}}
\subseteq \{1, \ldots , n\}$,  
one has no information about $\vec{s}$ if
$\# \mathcal{I} \le t$ 
 and 
one can recover $\vec{s}$ if
$\# \mathcal{I} \ge r$ \cite{Chen}. We shall always assume that $t$ is
largest possible and that 
$r$ is smallest possible such that the above hold. We say that one has a $t$-threshold secret sharing scheme if $t = r+1$. 
 
We consider the secret sharing schemes introduced in \cite[Section 4.2]{Chen}, which was the first general construction of 
 ramp secret sharing schemes using arbitrary linear codes:  Let 	
 $C_2 \subsetneq C_1 \subseteq \ff^n$ be two linear codes. Set $k_2 =
 \dim (C_2)$ and $k_1 = \dim (C_1)$ and let $L \subsetneq \ff^n$ be
 such that $C_1 = L  \oplus C_2$ (direct sum). That is, $L \cap
 C_2=\{\vec{0}\}$ and the union of a basis for $L$ and a basis for
 $C_2$ constitutes a basis for $C_1$. We denote by $\ell = \dim (L) = \dim (C_1 / C_2) =k_1 - k_2$.

We consider a secret $\vec{s} \in \ff^\ell$; note that $\ell >0$ since $C_1
\neq C_2$. We fix a vector space isomorphism $\psi : \ff^\ell \to L$ which maps
the secret $\vec{s} \in \ff^\ell$ to $L$, and choose $\vec{c}_2 \in C_2$ randomly
(uniformly distributed). Finally, consider $\vec{x} = \psi (\vec{s}) + \vec{c}_2 \in
C_1$. The $n$ shares consist of the $n$ coordinates of $\vec{x}$; this
scheme is clearly $\ff$-linear \cite{Chen}. One may also consider that
the secret $\vec{s}$ is represented by the coset $\psi (\vec{s}) + C_2$ in $C_1 /
C_2$. Note that there are $q^\ell$ different cosets in $C_1 / C_2$ and
there are $q^{k_2}$ possible representatives for every coset, i.e. for
generating the shares of a secret $\vec{s}$.  The schemes  in
\cite{MR2588125,Massey} form a particular case of the above scheme with $\ell = 1$.
\begin{remark}
All linear ramp secret sharing schemes with shares in ${\mathbb{F}}_q$ are of the above type.
For
constructions that use puncturing~\cite{Massey}, \cite[Sec.\ 4.1]{Chen} we can take $C_1,
C_2$ to be the punctured codes.
\end{remark}
Let $\mathcal{I} \subseteq \mathcal{J}  = \{ 1, \ldots, n\}$. We
consider that an unauthorized set of participants obtains the shares
$\{ x_i \mid i \in \mathcal{I}\}$. We  represent the shares by a random
variable $\vec{X}$, and the shares obtained by an unauthorized set of
participants by $ f_\mathcal{I}(\vec{x})=(x_i \mid i \in \mathcal{I})$
where $f_\mathcal{I} : \ff^n \to \ff^{\# \mathcal{I}}$. The amount of
information in $q$-bits that the unauthorized set obtains is measured by $I ( \vec{S} ; f_\mathcal{I} (\vec{X}))$,  
the mutual information, where $\vec{S}$ is the random variable that
represents the secrets, and $f_\mathcal{I}(\vec{X})$ is the  random
variable that represents the shares that an unauthorized set may obtain. We assume that both  $\vec{S}$ and $\vec{X}$ 
 are uniformly distributed. In particular we have $t$-privacy and $r$-reconstruction if $t$ is largest possible and $r$ is smallest
 possible such that 
$I ( \vec{S} ; f_\mathcal{I} (\vec{X}) )= 0$ for all $\#\mathcal{I}
 \le t$
and $I ( \vec{S} ; f_\mathcal{I} (\vec{X}) )= \ell$
 for all $\# \mathcal{I} \ge r$. A (non sharp) bound for $r$ and $t$
 was given in \cite{Chen}: $r < n - d(C_1)$ and $t > d(C_2^\perp)$ 
where $d(C_i)$ denotes the minimum distance of $C_i$, for $i=1,2$. The exact values  can be derived from \cite[Proof of Theorem 4]{kurihara} as
\begin{eqnarray}
I ( \vec{S} ; f_\mathcal{I} (\vec{X})) &=& \ell - \dim ( (V_{\overline{\mathcal{I}}} \cap C_1)/(V_{\overline{\mathcal{I}}} \cap C_2)),\label{eq:MI0}\\
 &=&  \dim ( (C_2^\perp \cap V_{\mathcal{I}} )/( C_1^\perp \cap V_{\mathcal{I}})),\label{eq:MI}
\end{eqnarray}where  $\overline{\mathcal{I}} = \mathcal{J} \setminus
I$ and $V_\mathcal{I} = \{ \vec{x} \in \ff^n \mid  x_i =0 {\mbox{ for
    all }} i \notin \mathcal{I} \}$.

For the convenience of the reader we include the computation of the
previous mutual information: since the
variables $\vec{S}$ and $\vec{X}$ are uniformly distributed one has
that $H_q (f_\mathcal{I}(\vec{X})) = \log_q \# f_\mathcal{I}(C_1)  =
\dim (f_\mathcal{I}(C_1)) = k_1 - \dim (\ker (f_\mathcal{I}) \cap
C_1)$, and  $H_q (f_\mathcal{I}(\vec{X}) | S) = \log_q \# f_\mathcal{I}( C_2 )  = \dim (f_\mathcal{I}(C_2)) = k_2 - \dim (\ker
(f_\mathcal{I}) \cap C_2)$. Here, $H_q$ is the entropy function to
base $q$. Therefore $I ( \vec{S} ; f_\mathcal{I} (\vec{X})) =
k_1 - k_2 - \big(\dim (\ker (f_\mathcal{I}) \cap C_1)- \dim (\ker (f_\mathcal{I})
\cap C_2)\big)$ and we obtain equation (\ref{eq:MI0}). Equation
(\ref{eq:MI}) follows from (\ref{eq:MI0}) and an extension of
Forney's second duality lemma \cite[Lemma 25]{Matsumoto2}: Let $V
\subseteq \ff^n$, then 
\begin{equation*}\label{eq:forney}\dim ( (C_2^\perp \cap V^\perp)/(C_1^\perp \cap V^\perp) ) = \dim (C_1 / C_2) - \dim ((C_1 \cap V)/(C_2 \cap V)).\end{equation*}

In order to characterize the security of secret sharing schemes,  one considers the
$j$th relative dimension/length profile (RDLP) of two codes $C_2
\subsetneq C_1$ with $j \in \{1, \ldots, n\}$ \cite{luo}: $$K_j ( C_1,
C_2) = \max_{
{\mathcal{I}} \subseteq {\mathcal{J}}, 
\# \mathcal{I} = j} \dim ( (C_1 \cap V_{\mathcal{I}} )/( C_2 \cap V_{\mathcal{I}})),$$and the $m$th relative generalized Hamming weight  (RGHW) with $m \in \{1, \ldots, \ell\}$  \cite{luo}:
\begin{equation}
M_m ( C_1, C_2) = \min_{\mathcal{I} \subseteq \mathcal{J}} \{
\# \mathcal{I} \mid \dim ( (C_1  \cap V_{\mathcal{I}} )/( C_2  \cap
V_{\mathcal{I}})) = m \}. \label{eqluo}
\end{equation}

In this way the worst amount of information leakage  of $\vec{s}$
from $j$ shares is precisely characterized by the $j$th relative
dimension/length profile of $C_2^\perp$ and $C_1^\perp$ \cite[Theorem
  4]{kurihara}: 
$$\max_{ {\mathcal{I}} \subseteq {\mathcal{J}}, \# \mathcal{I} = j } I ( \vec{S} ;
f_\mathcal{I} (\vec{X})) = \max_{ {\mathcal{I}} \subseteq
  {\mathcal{J}}, \# \mathcal{I} = j } \dim
( (C_2^\perp \cap V_{\mathcal{I}} )/( C_1^\perp \cap V_{\mathcal{I}})) =
K_j ( C_2^\perp, C_1^\perp).$$
The smallest possible number of shares
for which an unauthorized set of participants can determine $m$ $q$-bits of information is
\begin{eqnarray*}\min_{\mathcal{I} \subseteq \mathcal{J} }
  \{\# \mathcal{I}  \mid  I ( \vec{S} ; f_\mathcal{I} (\vec{X})) =m \} & =
  &  \min_{\mathcal{I} \subseteq \mathcal{J} } \{ \# \mathcal{I} \mid  \dim
  ( (C_2^\perp \cap V_{\mathcal{I}} )/( C_1^\perp \cap V_{\mathcal{I}}))
  = m \}\\ & = & M_m ( C_2^\perp, C_1^\perp).\end{eqnarray*}
In particular $t=M_1 ( C_2^\perp, C_1^\perp) -1$ \cite[Theorem
  9]{kurihara}. (See also~\cite[Th.\ 6.7]{Bains} and for the special case of $\ell=1$ \cite[Cor.\ 1.7]{MR2588125}). We now generalize the notion of $t$-privacy and $r$-reconstruction.
\begin{definition}\label{def:privreco}
We say that a ramp secret sharing scheme has
$(t_1,\ldots,t_\ell)$-privacy and $(r_1, \ldots ,
r_\ell)$-reconstruction if $t_1,\ldots , t_{\ell}$ are chosen largest
possible and $r_1, \ldots
,r_{\ell}$ are chosen smallest possible such that:
\begin{itemize}
\item an adversary cannot obtain $m$ $q$-bits of information about $\vec{s}$ with any
  $t_m$ shares, 
\item it is possible to recover $m$ $q$-bits of information about $\vec{s}$ with any
  collection of $r_m$ shares. 
\end{itemize}
In particular, one has $t=t_1$ and $r = r_\ell$.
\end{definition}

By our previous discussion one has that $t_m = M_{m} (C_2^\perp,
C_1^\perp) -1$ since $M_{m} (C_2^\perp, C_1^\perp)$ is the smallest size
of a set of shares that can determine $m$ $q$-bits of information about $\vec{s}$ \cite[Theorem 4]{kurihara}. We will show that $(r_1, \ldots, r_\ell)$ can be characterized in
terms of the RGHWs as well. Let $r'_m$ be the  largest size of a set of shares that
cannot determine $m$ $q$-bits of information about $\vec{s}$, i.e.
\begin{equation}\label{eq:rp}
r'_m =  \max_{\mathcal{I} \subseteq \mathcal{J} } \{\#\mathcal{I}  \mid  I ( \vec{S}; f (\vec{X})) < m  \}.\end{equation}
 
This value is closely related to $r_m$ since any strictly larger set
of shares will determine $m$ $q$-bits of information about $\vec{s}$ and thus
\begin{eqnarray}
r_m & = & r'_{m} + 1\notag\\
  & = & \max_{\mathcal{I} \subseteq \mathcal{J} } \{\# \mathcal{I}  \mid  I ( \vec{S}; f_\mathcal{I}  (\vec{X})) < m   \} +1\notag\\
  & = & \max_{\mathcal{I} \subseteq \mathcal{J} } \{\# \mathcal{I}  \mid  I ( \vec{S}; f_\mathcal{I}  (\vec{X})) =m - 1 \} +1\notag\\
& = & \max_{\mathcal{I} \subseteq \mathcal{J} } \{ \# \mathcal{I} \mid   \dim ((C_1 \cap V_{\overline{\mathcal{I}}})/(C_2 \cap V_{\overline{\mathcal{I}}})) = \ell - m + 1 \} +1\textrm{,~by~(\ref{eq:MI0})}\notag\\
& =  & n - \min_{\overline{\mathcal{I}} \subseteq \mathcal{J} } \{ \#\overline{\mathcal{I}} \mid  \dim ((C_1 \cap V_{\overline{\mathcal{I}}})/(C_2 \cap V_{\overline{\mathcal{I}}})) =   \ell - m + 1 \} +1\notag\\
& = & n - M_{\ell-m +1}(C_1,C_2)+1. \label{eq:rm}
\end{eqnarray}
In particular one has that $r = r_\ell = n - M_1 ( C_1, C_2)
+1$  \cite[Theorem 9]{kurihara} (see
also~\cite[Cor.\ 1.7]{MR2588125} for the special case $\ell=1$).	We note that $r'_m$ corresponds to the $(m-1)$th  conjugate relative length/dimension profile in \cite{Luo2}.

\begin{theorem}\label{theconnection}
Let $C_1/C_2$, where $\dim ( C_1)-\dim (C_2)=\ell$, be a linear ramp secret
sharing scheme with $(t_1, \ldots , t_\ell)$-privacy and $(r_1, \ldots
, r_\ell)$-reconstruction. Then for $m=1, \ldots , \ell$ we have
$t_m=M_m(C_2^\perp,C_1^\perp)-1$ and $r_m=n-M_{\ell -m+1}(C_1,C_2)+1$.
\end{theorem}

We shall relate the above concept
of $(t_1, \ldots , t_\ell)$-privacy and $(r_1, \ldots
, r_\ell)$-reconstruction
to the literature: let $D_1 \subsetneq D_2 \subseteq \mathbb{F}_q^n$
be vector spaces of codimension $\ell$ and define  for $1 \leq m \leq \ell$,
$$
A_m(D_1,D_2)  = \{ {\mathcal{I}} \subseteq {\mathcal{J}} \mid m = \dim(D_1 \cap
V_{\mathcal{I}})/(D_2 \cap V_{\mathcal{I}})  \}.
$$ 
Since $I ( \vec{S} ; f_\mathcal{I} (\vec{X}))  = \dim ( (C_2^\perp
\cap V_{\mathcal{I}} )/( C_1^\perp \cap V_{\mathcal{I}}))$ we have
that, for $D_1 = C_2^\perp$ and $D_2 = C_1^\perp$, $A_m(D_1,D_2)$ is the
collection of shares that give $m$ $q$-bits of information about
$\vec{S}$. In addition, $A_{\ell}(D_1,D_2)$ is the access structure in the sense of
\cite{Ito}, and $A_m(D_1,D_2)$ is equivalent to $A_m$ in~\cite[Eq.\ (3.1)]{Iwamoto}.

In particular we are interested in the largest and smallest element
of 	such a collection of shares 
$$
\begin{array}{c}
A_m^{\mathrm{min}}(D_1,D_2) = \{ {\mathcal{I}} \in A_m(D_1,D_2) \mid \nexists {\mathcal{K}} \in A_m(D_1,D_2) \textrm{ s.t. } {\mathcal{K}} \subsetneq {\mathcal{I}} \}\\
A_m^{\mathrm{max}}(D_1,D_2) = \{ {\mathcal{I}} \in A_m(D_1,D_2) \mid \nexists {\mathcal{K}} \in A_m(D_1,D_2) \textrm{ s.t. } {\mathcal{K}} \supsetneq {\mathcal{I}} \}
\end{array}
$$
and, as we are interested in its size, we define
$$
\begin{array}{c}
A_m^d(D_1,D_2) = \{ {\mathcal{I}} \in A_m(D_1,D_2) \mid d=\# {\mathcal{I}} \}\\
A_m^{\mathrm{min},d}(D_1,D_2) = \{ {\mathcal{I}} \in A_m^{\mathrm{min}}(D_1,D_2)
\mid d = \# {\mathcal{I}} \}\\
A_m^{\mathrm{max},d}(D_1,D_2) =  \{ {\mathcal{I}} \in
A_m^{\mathrm{max}}(D_1,D_2) \mid d = \# {\mathcal{I}} \}.
\end{array}
$$

Moreover, we are interested in the smallest and
the largest size of a
collection of shares that reveal $m$ $q$-bits of information: 
the first one being the smallest $d \in \{ 1, \ldots , n\}$ such that $ A_m^{\mathrm{min},d}(D_1,D_2)$ is non-empty and it is equal to $M_m (D_1,D_2) = t_m +1$. Analogously, the largest size of a collection of shares that reveals $m$ $q$-bits of information is the largest $d \in \{1, \ldots , n\}$ such that $A_m^{\mathrm{max},d}(D_1,D_2)$ is non-empty and it is equal to $n - M_{\ell-m+1} (C_1,C_2) +1 = r_m$.

Ramp secret sharing schemes with $\ell >1$ are relevant in the situation
where the set of possible secrets is large but one wants to keep the
size of each share small. A further motivation for considering $\ell >1$
is the analogy to the wiretap channels of type II
\cite{wyner,ozarow}. Recall that this model involves a main channel from
Alice to Bob which is assumed to be error and erasure free, and a
secondary channel from Alice to the eavesdropper Eve which is a
$q$-ary erasure channel. Consider the slightly more general situation
where also the main channel is a $q$-ary erasure channel
\cite{subramanian2009mds}. Assuming that the probability of erasure is much smaller on
the main channel than on the secondary channel we see that to achieve
reliable and secure communication we should use long codes $C_2
\subsetneq C_1$. To retain a positive information rate on the main
channel we therefore need $\ell >1$. The exact values of the mutual
information on the main and the secondary channel could be calculated from
$A_m(D_1, D_2)$, $m=1, \ldots , \ell$ and the  erasure probabilities of
the two channels; but it seems a difficult task to determine
$A_m(D_1,D_2)$ even for simple codes. Finding $M_m (D_1,D_2) = t_m +1$ and
$n - M_{\ell-m+1} (C_1,C_2) +1 = r_m$,
however, 
would be a first step in this direction. As we shall see in the
following, for many codes we can easily estimate these last mentioned parameters.\\

In the remaining part of this paper we shall concentrate on
methods to estimate RGHW. We shall need the following definition which
by~\cite{luoetal} is equivalent to~(\ref{eqluo}) (see
also~\cite[Def.\ 6.2]{Bains}).
\begin{definition}\label{defrelham}
Let $C_2 \subsetneq C_1 $ be linear codes over ${\mathbb{F}}_q$. For $m
=1,\ldots , \dim ( C_1)-\dim (C_2)$ the $m$th relative
generalized Hamming weight is defined as
\begin{eqnarray}
&&M_m(C_1,C_2)\nonumber \\
& &= \min \{ \# {\mbox{Supp}}\, D \mid D {\mbox{ \ is a
    subspace of \ }} C_1, \dim (D)=m, D \cap
C_2 =\{ \vec{0} \} \}.\nonumber
\end{eqnarray} 
\end{definition}

From this definition the connection between the RGHW and
the 
generalized Hamming weight (GHW) becomes clear -- the latter being
$d_m(C_1)=M_m(C_1,C_2)$ with $C_2=\{\vec{0}\}$. Before embarking with more general classes of codes in the
next section we discuss the parameters $t_m,r_m$ in the case of
MDS codes.

\section{Ramp schemes based on MDS codes}\label{sec3}
Let $C$ be an MDS code of dimension $k$. Then $C^\perp$ is also MDS
and consequently 
\begin{eqnarray}
d_m(C)=n-k+m,&&m=1, \ldots , k \label{eqhyp1} \\
d_m(C^\perp)=k+m,&&m=1, \ldots, n-k \label{eqhyp2}
\end{eqnarray}
which means that all generalized Hamming weights attain the Singleton bound. 
Consider two MDS codes $C_2 \subsetneq C_1$ with $\dim (C_1)=k_1$
and $\dim (C_2) = k_2$. By definition, $M_m(C_1,C_2)\geq d_m(C_1)$,
$m=1, \ldots ,\ell= k_1-k_2$. However, the
Singleton bound for RGHW is identical
to the Singleton bound for GHW~\cite[Sec.\ IV]{luo} and therefore 
$M_m(C_1,C_2)=d_m(C_1)$ and $M_m(C_2^\perp,C_1^\perp)=d_m(C^\perp_2)$
\cite{subramanian2009mds}. Based on~(\ref{eqhyp1}) and
(\ref{eqhyp2}) one can show that
\begin{eqnarray}
M_m(C_2^\perp,C_1^\perp)=n-M_{\ell-m+1}(C_1, C_2)+1, \label{eqhyp3}
\end{eqnarray}
and from Theorem~\ref{theconnection} it now follows that if we base a ramp scheme on two MDS codes 
then 
the size of a group uniquely
determines how much information it can reveal:
$$t_m=r_m-1, \, \, \, t_{m+1}=t_m+1, \, \, \, t_1=k_2, \, \, \, r_\ell=k_1.$$
When the number of
participants is larger than two times the field size minus 1 then
by~\cite[Cor.\ 7.4.4]{huffman2003fundamentals} 
$C_1$ and $C_2$ cannot be MDS -- unless $k_1=n-1$ and $k_2=1$ -- and consequently we
can no longer assume~(\ref{eqhyp3}). What is obviously needed is a method to
estimate the left and the right side of~(\ref{eqhyp3}) for codes of any length. As shall be demonstrated in the
following the Feng-Rao method makes this possible.
\section{The Feng-Rao bounds for RGHW}\label{sec4}

The Feng-Rao bounds come in two versions: One for primary codes~\cite{AG,MR2831617,agismm}
and one for dual codes~\cite{FR24,FR1,FR2,pellikaan93efficient,handbook,MM}. The most general formulations deal
with arbitrary linear codes, whereas more specialized formulations -- such
as the order bounds -- require that the code construction is supported
by certain types of algebraic structures. The
bounds have been applied to the minimum distance, the generalized Hamming weights -- and for the case of dual
codes of co-dimension $1$ -- also the relative minimum distance~\cite{MR2588125}. It is not difficult to extend the method for estimating GHW to a method
for estimating RGHW. In the following we give the details for primary
codes in the language of general linear codes. The details for dual
codes are similar, hence for these codes we shall give a more
brief description.\\

\noindent We start by introducing some terminology that shall be used
throughout the section. Let ${\mathcal{B}}=\{\vec{b}_1, \ldots ,
\vec{b}_n\}$ be a fixed basis for ${\mathbb{F}}_q^n$ as a vector space over
${\mathbb{F}}_q$ and write ${\mathcal{J}}=\{1, \ldots , n\}$. 
\begin{definition}
The function $\bar{\rho} : {\mathbb{F}}_q^n \rightarrow \mathcal{J}
\cup \{ 0 \}$ is given as follows. For non-zero $\vec{c}$ we have
$\bar{\rho}(\vec{c})=i$ where $i$ is the unique integer such that
$$\vec{c} \in {\mbox{Span}} \{ \vec{b}_1, \ldots , \vec{b}_i\}
\backslash {\mbox{Span}} \{ \vec{b}_1, \ldots ,
\vec{b}_{i-1}\}.$$ Here we used the convention that ${\mbox{Span}} \, 
\emptyset = \{ \vec{0}\}$. Finally, $\bar{\rho}(\vec{0})=0$.
\end{definition}
\noindent The component wise product of two vectors in
${\mathbb{F}}_q^n$ plays a fundamental role in our exposition. This
product is given by $$(\alpha_1, \ldots , \alpha_n) \ast (\beta_1, \ldots , \beta_n)=(\alpha_1\beta_1,
\ldots , \alpha_n\beta_n).$$ 
\begin{definition}
An ordered pair $(i,j) \in {\mathcal{J}} \times {\mathcal{J}}$ is said to be
one-way well-behaving (OWB) if $\bar{\rho}(\vec{b}_{i^\prime} \ast
\vec{b}_j)< \bar{\rho}(\vec{b}_{i} \ast \vec{b}_j)$ holds true for all
$i^\prime \in {\mathcal{J}}$ with 
$i^\prime < i$.  
\end{definition}
\begin{definition}\label{defdif}
For $i \in {\mathcal{J}}$  define 
$$\Lambda_i=\{ l \in {\mathcal{J}}
\mid \exists \, j \in {\mathcal{J}} {\mbox{ such that }} (i,j) {\mbox{
    is OWB and }} \bar{\rho}(\vec{b}_i \ast \vec{b}_j)=l \}.$$ 
\end{definition}
\noindent As is easily seen -- if $D \subseteq {\mathbb{F}}_q^n$ is a vector space of dimension $m$
then it holds that 
 $\# \bar{\rho}\big( D \backslash \{ \vec{0}\} \big)=m$. 
(Actually, any set 
$\{\vec{d}_1, \ldots , \vec{d}_m\}\subseteq D\backslash \{ \vec{0}\}$ with
$\bar{\rho}(\vec{d}_1)< \cdots <\bar{\rho}(\vec{d}_m)$ constitutes a
basis for $D$).
The following result is a slight modification of the material in~\cite{AG}.
\begin{proposition}\label{prothe}\label{propthe}
Let $D \subseteq {\mathbb{F}}_q^n$ be a vector space of dimension at
least $1$. The support size of $D$ satisfies
\begin{equation}
\# {\mbox{Supp}} (D) \geq \# \cup_{i \in \bar{\rho}(D \backslash \{
    \vec{0} \})} \Lambda_i. \label{eqbound}
\end{equation}
\end{proposition}
\begin{proof}
Let $l_1 < \cdots < l_\sigma$ be the
elements in $ \cup_{i \in \bar{\rho}(D \backslash \{
    \vec{0} \})} \Lambda_i$ and let $i_1, \ldots , i_\sigma$ and
$j_1, \ldots , j_\sigma $ be such that for $s=1, \ldots , \sigma$ it
holds that:
\begin{itemize}
\item $i_s \in \bar{\rho}(D \backslash \{ \vec{0} \})$,
\item $(i_s,j_s)$ is OWB and $\bar{\rho}(\vec{b}_{i_s} \ast \vec{b}_{j_s})=l_s$.
\end{itemize} 
Choose $\vec{d}_1, \ldots \vec{d}_\sigma \in D$ 
with
$\bar{\rho}(\vec{d}_s)=i_s$, $s=1, \ldots , \sigma$. Clearly
$\bar{\rho}(\vec{d}_s \ast \vec{b}_{j_s})=l_s$ and therefore
$\vec{d}_1 \ast \vec{b}_{j_1}, \ldots , \vec{d}_{\sigma}\ast
\vec{b}_{j_\sigma}$ are linearly independent. In conclusion $D \ast
    {\mathbb{F}}_q^n= \{\vec{d} \ast \vec{c} \mid \vec{d} \in D,
      \vec{c} \in {\mathbb{F}}_q^n\}$ is of dimension at least
      $\sigma$. The
      dimension of $D \ast {\mathbb{F}}_q^n$ equals the size of the
      support of $D$ and the proposition follows. 
\end{proof}

\noindent We now turn to RGHW. Observe that although $C_2 \subsetneq C_1$ implies
 $\bar{\rho}(C_2) \subsetneq \bar{\rho}(C_1)$,
it does not always hold that 
$\vec{c} \in C_1 \backslash C_2$ implies $\bar{\rho}(\vec{c}) \in \bar{\rho}(C_1)\backslash
\bar{\rho}(C_2)$. However, some observations can still be made.
\begin{theorem}\label{thethat}
Consider linear codes $C_2 \subsetneq C_1$, $\dim (C_1)=k_1$, $\dim (C_2)=k_2$. Let $u$ be the smallest element in $\bar{\rho}(C_1)$ that is not in
  $\bar{\rho}(C_2)$. For $m=1, \ldots , k_1 - k_2$ we have
\begin{eqnarray}
M_m(C_1,C_2)&\geq &\min \big\{ \# \cup_{s=1}^m \Lambda_{i_s} \mid u\leq
i_1 < \cdots < i_m, \nonumber \\
&&\, \, \, \, \, \, \, \, \, \, \, \, \, \, \, \, \, \, \, \, \, \, \, \, \, \, \, \, \, \, \, \, \, \, \, \, \, \, \, \, \, \, \, \, \, \, \, \, \, \, \, \, \, \, 
i_1, \ldots ,
i_m \in \bar{\rho}(C_1 \backslash \{\vec{0}\})\big\}. \nonumber
\end{eqnarray}
\end{theorem}
\begin{proof}
If $D$ is an $m$-dimensional subspace of $C_1$ with $D \cap
C_2=\{\vec{0} \}$ then we can write $\bar{\rho}(D\backslash \{ \vec{0}\})=\{i_1, \ldots , i_m\} \subseteq
\bar{\rho}(C_1 \backslash \{ \vec{0}  \})$ with $u
\leq i_1 < \cdots < i_m$. The theorem now follows from Proposition~\ref{prothe}.
\end{proof}
\begin{corollary}\label{corcor}
Consider a $k_1$-dimensional code $C_1$, say $C_1={\mbox{Span}} \{ \vec{f}_1, \ldots ,
\vec{f}_{k_1}\}$, where without loss of generality we assume  $\bar{\rho}(\vec{f}_1) < \cdots <
\bar{\rho}(\vec{f}_{k_1})$. For $k_2 < k_1$ let $C_2={\mbox{Span}}\{
\vec{f}_1, \ldots , \vec{f}_{k_2} \}$. We have
\begin{eqnarray}
M_m(C_1,C_2)&\geq &\min \big\{ \# \cup_{s=1}^m \Lambda_{i_s} \mid i_1 <
\cdots < i_m, \nonumber \\
&&\, \, \,  \, \, \,  \, \, \, i_1, \ldots , i_m \in \{
\bar{\rho}(\vec{f}_{k_2+1}), \ldots , \bar{\rho}(\vec{f}_{k_1})\} \big\}.
\end{eqnarray}
\end{corollary}
Next we treat dual codes.
\begin{definition}
For $\vec{c} \in {\mathbb{F}}_q^n \backslash \{ \vec{0}\}$ define
$M(\vec{c})$ to be the smallest number $i \in {\mathcal{J}}$ such that
$\vec{c} \cdot \vec{b}_i \neq 0$. Here $\vec{a} \cdot \vec{b}$ means the
usual inner product between $\vec{a}$ and $\vec{b}$.
\end{definition}
It is clear that for an $m$-dimensional space $D$ we have $\#
M(D\backslash \{\vec{0}\})=m$. Also it is clear that if $D \subseteq
C^\perp$, where $C$ is a linear code, then $M(D\backslash \{\vec{0}\})
\cap \bar{\rho}(C) = \emptyset$.
\begin{definition} For $l \in {\mathcal{J}}$ define
$$V_l=\{i \in {\mathcal{J}} \mid \bar{\rho}(\vec{b}_i \ast \vec{b}_j)=l
  {\mbox{ for some }}\vec{b}_j \in {\mathcal{B}} {\mbox{ with }} (i,j)
  {\mbox{ OWB}}\}.$$
\end{definition}
The following result is proved by slightly modifying the proof
of~\cite[Prop.\ 3.12]{heijnenpellikaan} and \cite[Th.\ 5]{geithom}. 
\begin{proposition}
Let $D \subseteq {\mathbb{F}}_q^n$ be a space of dimension at least
$1$. We have 
$$\# {\mbox{Supp}} (D) \geq \# \cup_{l \in M(D\backslash \{\vec{0}\})} V_l.$$
\end{proposition}
From the above discussion we derive
\begin{theorem}\label{theny}
Consider linear codes $C_2 \subsetneq C_1$. Let $u$ be the largest
element in $\bar{\rho}(C_1 \backslash \{\vec{0}\})$. For $m=1, \ldots
, \dim (C_1)-\dim (C_2)=\dim (C_2^\perp) -\dim (C_1^\perp)$ we have
\begin{eqnarray}
M_m(C_2^\perp,  C_1^\perp) &\geq &\min \{ \#  \cup_{s=1}^m V_{i_s}
\mid 1 \leq i_1 < \cdots <i_m \leq u, \nonumber \\
&&{\mbox{ \ \ \ \ \ \ \ \ \ \ \ \ \ \ \ \ \ \ \ \ \  \ \ \ \ \ \ \ }} i_1, \ldots , i_m \notin \bar{\rho}(C_2)\}.
\end{eqnarray}
\end{theorem}

\noindent To apply Theorem~\ref{thethat},
Corollary~\ref{corcor} and Theorem~\ref{theny} we need information on which 
pairs are OWB. This suggests the use of a supporting algebra. One
class of algebras that works well is the order
domains~\cite{handbook,MR1826339,GP}. In the present paper we will 
concentrate on the most prominent example of order domain codes
--  namely one-point algebraic geometric codes.
\begin{remark}
In our exposition we used a single (but arbitrary) basis ${\mathcal{B}}$ for
${\mathbb{F}}_q^n$ as a vector space over
${\mathbb{F}}_q$. Following~\cite{pellikaan93efficient} one could reformulate all the
above results in a more general setting that uses three bases ${\mathcal{U}}$,
${\mathcal{V}}$, and ${\mathcal{W}}$. This point of view is important
when one considers affine variety codes~\cite{salazar}, but it does not improve the
results for order domain codes. In~\cite{geilmartin2013further} and \cite{geil2013improvement}, the concept of
OWB was relaxed giving new improved Feng-Rao bounds. All the above
results could be reformulated in this setting -- but again -- for
order domain codes the results stay unchanged. 
\end{remark}

\section{One-point algebraic geometric codes}\label{sec5}
Given an algebraic function field $F$ of transcendence degree one, let
$P_1, \ldots , P_n$, $Q$ be distinct rational places. For $f \in F$ write
$\rho(f)=-\nu_Q(f)$, where $\nu_Q$ is the valuation at $Q$, and denote by $H(Q)$ the Weierstrass semigroup of
 $Q$. That is, $H(Q)=\rho\big( \cup_{\mu=0}^\infty
{\mathcal{L}}(\mu Q) \big)$. In the following let $\{f_{\lambda} \mid
\lambda \in H(Q)\}$ be any fixed basis for $R=\cup_{\mu=0}^\infty
        {\mathcal{L}}(\mu Q)$ with $\rho(f_\lambda)=\lambda$ for all
        $\lambda \in H(Q)$. Let 
$D=P_1+\cdots +P_n$ and define
\begin{eqnarray}
H^\ast (Q)&=&\{ \mu \mid C_{\mathcal{L}}(D,\mu Q) \neq
C_{\mathcal{L}}(D,(\mu-1) Q)\} \nonumber \\
&=&\{\gamma_1, \ldots , \gamma_n\} \subsetneq H(Q). \label{eqenting} 
\end{eqnarray}
Here, the enumeration is chosen such that $\gamma_1 < \cdots < \gamma_n$.
 Consider the map ${\mbox{ev}}: F \rightarrow {\mathbb{F}}_q^n$ given
 by ${\mbox{ev}}(f)=(f(P_1), \ldots , f(P_n))$. 
The set  
\begin{equation}
\{ \vec{b}_1={\mbox{ev}}(f_{\gamma_1}), \ldots ,
        \vec{b}_n={\mbox{ev}}(f_{\gamma_n})\} \label{eqhoejsa}
\end{equation}
clearly is a basis for ${\mathbb{F}}_q^n$ and by~\cite[Pro.\ 27]{AG} a pair $(i,j)$ is OWB if $\rho (f_{\gamma_i })+\rho(f_{\gamma_j}) =
\rho (f_{\gamma_l})$, i.\ e.\ $\gamma_i+\gamma_j=\gamma_l$, in which case of course $\bar{\rho} (\vec{b}_i \ast
\vec{b}_j)=l$. From~\cite[Pro.\ 28]{AG} we know that if
$\delta \in H^\ast (Q)$ and $\alpha, \beta \in H(Q)$ satisfy
$\alpha+\beta=\delta$ then we have  $\alpha , \beta \in H^\ast (Q)$. We therefore get the following lemma.
\begin{lemma}\label{proptraekker}
Let $\{\vec{b}_1, \ldots , \vec{b}_n\}$ be as above. For $i \in
{\mathcal{J}}$ it holds that
$$
\{ l \in {\mathcal{J}} \mid \gamma_l-\gamma_i \in H(Q)\} \subseteq
\Lambda_i$$
where $\Lambda_i$ is as in Definition~\ref{defdif}.
\end{lemma}

\begin{proposition}\label{the1}
Let $D \subseteq {\mathbb{F}}_q^n$ be a vector space of dimension
$m$. There exist unique numbers $\gamma_{i_1}< \cdots < \gamma_{i_m}$
in $H^\ast(Q)$ such that $\bar{\rho}(D \backslash \{\vec{0}\})=\{i_1,
\ldots , i_m\}$. The support of $D$ satisfies
\begin{eqnarray}
\# {\mbox{Supp}}(D)&\geq & \# \bigg( H^\ast (Q) \cap
\big( \cup_{s=1}^m(\gamma_{i_s}+H(Q))\big) \bigg) \label{eqdej}\\
&\geq& n-\gamma_{i_m} + \# \{\lambda \in \cup_{s=1}^{m-1}
(\gamma_{i_s}+H(Q)) \mid \lambda \notin \gamma_{i_m}+H(Q)\}. \label{eqihopla2}
\end{eqnarray}
\end{proposition}
\begin{proof}
By Lemma~\ref{proptraekker} the
right side of~(\ref{eqdej}) is lower than or equal to $\#
\cup_{s=1}^m\Lambda_{i_s}$, and (\ref{eqdej}) therefore 
follows from Proposition~\ref{propthe}. Another way of writing the right
side of (\ref{eqdej}) is $n-\#\big( H^\ast(Q) \backslash
\cup_{s=1}^m(\gamma_{i_s}+H(Q))\big)$. This number is greater than or
equal to 
\begin{eqnarray}
&&n-\#\big( H(Q) \backslash
\cup_{s=1}^m(\gamma_{i_s}+H(Q))\big)\nonumber \\
&=&n-\# \big(H(Q)\backslash(\gamma_{i_m}+H(Q))\big)\nonumber \\
&&{\mbox{ \ \ \ \ \ \ \ \ \ \ \ }}  + \# \,  \{ \lambda
\in \cup_{s=1}^{m-1}(\gamma_{i_s}+H(Q)) \mid \lambda \notin \gamma_{i_m}+H(Q)\}.\nonumber
\end{eqnarray}
From~\cite[Lem.\ 5.15]{handbook} we know that for any
numerical semigroup $\Gamma$ and $\lambda \in \Gamma$, one has
$\lambda=\# \big( \Gamma \backslash (\lambda+\Gamma) \big)$. In
particular $\# \big( H(Q) \backslash
(\gamma_{i_m}+H(Q))\big)=\gamma_{i_m}$ and (\ref{eqihopla2}) follows.
\end{proof}
 
From~(\ref{eqihopla2}) we can obtain a
manageable bound on the RGHWs of
one-point algebraic geometric codes as we now explain. This bound
can even be used when one does not know $H^\ast (Q)$. Given
non-negative integers $\lambda_1
< \cdots < \lambda_m$ (note that we make no assumptions that
$\lambda_1, \ldots , \lambda_m \in H(Q)$) let
$i_j=\lambda_j-\lambda_m$, $j=1, \ldots , m-1$
and observe that 
\begin{eqnarray}
\# \{ \lambda \in \cup_{s=1}^{m-1} (\lambda_i+H(Q) \mid \lambda
  \notin \lambda_m+H(Q)\}\, \, \, \, \, \, \, \, \, \, \, \, \, \, \, \nonumber \\
=\#\{ \alpha \in \cup_{s=1}^{m-1}(i_s+H(Q)) \mid \alpha \notin H(Q)\}\label{eqihopla1}
\end{eqnarray}
since $\lambda$ is in the first set if and only if $\lambda-\lambda_m$
is in the second set. The function $Z$ in the definition below shall
help us estimate the last expression in~(\ref{eqihopla2}).
\begin{definition}\label{defZ}
Consider a numerical semigroup $\Gamma$ and a positive integer
$\mu$. Define $Z(\Gamma,\mu,1)=0$ and for $1 < m \leq \mu$
\begin{eqnarray}
Z(\Gamma, \mu,m)&=&\min \big\{ \# \{\alpha \in \cup_{s=1}^{m-1}(i_s+\Gamma)
\mid \alpha \notin \Gamma\} \mid \nonumber \\
&&{\mbox { \ \ \ \ \ \ \ \ \ \ \ }} -\mu+1 \leq i_1 < \cdots < i_{m-1}
\leq -1\big\}.
\end{eqnarray}
\end{definition}
We are now ready for the main result of the section.
\begin{theorem}\label{the2}
Let $\mu_1, \mu_2$ be positive integers with $\mu_2 <\mu_1$.\\ 
For $m=1, \ldots ,
\dim (C_{\mathcal{L}}(D,\mu_1Q))-\dim (C_{\mathcal{L}}(D,\mu_2Q))$ we have
\begin{eqnarray}
&&M_m(C_{\mathcal{L}}(D,\mu_1Q),C_{\mathcal{L}}(D,\mu_2Q)) \nonumber \\
&\geq & \min \bigg\{ \# \big( H^\ast (Q) \cap
\big( \cup_{s=1}^m(\gamma_{i_s}+H(Q))\big) \big)\nonumber \\
&&{\mbox{ \ \ \ \ }} \mid \gamma_{i_1},
\ldots , \gamma_{i_m} \in H^\ast(Q), \mu_2<\gamma_{i_1} < \cdots <
\gamma_{i_t} \leq \mu_1\bigg\} \label{eqdejlig}\\
&\geq &\min \bigg\{n -\gamma_{i_m} + \# \{\lambda \in \cup_{s=1}^{m-1}
(\gamma_{i_s}+H(Q)) \mid \lambda \notin \gamma_{i_m}+H(Q)\} \nonumber \\
&& {\mbox{ \ \ \ \ }} \mid \gamma_{i_1},
\ldots , \gamma_{i_m} \in H^\ast(Q), \mu_2<\gamma_{i_1} < \cdots <
\gamma_{i_t} \leq \mu_1\bigg\} \label{eqihopla2lig}\\
 & \geq &n-\mu_1+Z(H(Q),\mu,m), \label{eqthebsound}
\end{eqnarray}
where $\mu=\mu_1-\mu_2$.
\end{theorem}
\begin{proof}
Consider an $m$-dimensional vector space $D \subseteq
C_{\mathcal{L}}(D,\mu_1Q)$ with $D\, \cap \, C_{\mathcal{L}}(D,\mu_2Q)=\{
\vec{0} \}$. Let $\gamma_{i_1} < \cdots < \gamma_{i_m}$ be as
described in Theorem~\ref{the1}. By the definition of the codes we have $\gamma_{i_1}, \ldots ,
\gamma_{i_m} \in \{ \mu_2+1,\ldots , \mu_1\}$ (this is the situation of Corollary~\ref{corcor}). Consequently  (\ref{eqdejlig}) and
(\ref{eqihopla2lig}), respectively, follow from (\ref{eqdej}) and
(\ref{eqihopla2}), respectively. We have $-\mu_1 \leq
-\gamma_{i_m}$. Similarly, by~(\ref{eqihopla1}) $Z(H(Q),\mu,m)$ is
smaller than 
or equal to the last term
in~(\ref{eqihopla2}). These observations prove (\ref{eqthebsound}).
\end{proof}

Note that~(\ref{eqthebsound}) may be strictly smaller
than~(\ref{eqihopla2lig}).  Firstly, $\mu_1$ may not 
belong to $H^\ast
(Q)$. Secondly, when applying the function $Z(H(Q),\mu,m)$ we do not discard the
numbers in $\{\mu_2+1, \ldots , \mu_1-1\}$ that are gaps of $H(Q)$,
and least of all the numbers in the interval that are not present in
$H^\ast (Q)$. The connection to the usual Goppa bound for primary
codes is seen from the expression in~(\ref{eqthebsound}): letting $m=1$
we get by Definition~\ref{defZ} $Z(H(Q),\mu,m)=0$ and the 
formula simplifies to the well-known bound on the minimum distance
$d\big(C_{\mathcal{L}}(D,\mu_1Q)\big) \geq n-\mu_1$.\\

For duals of one-point algebraic geometric codes we have
a bound similar to~(\ref{eqdejlig}), but no bounds similar to (\ref{eqihopla2lig}) or (\ref{eqthebsound}).
\begin{theorem}\label{the2dual}
Let $\mu_1, \mu_2$ and $m$ be as in Theorem~\ref{the2}. We have
\begin{eqnarray}
&&M_m(C_{\mathcal{L}}^\perp(D,\mu_2Q),C_{\mathcal{L}}^\perp(D,\mu_1Q)) \nonumber \\
&\geq & \min \bigg\{ \# \big( H (Q) \cap
\big( \cup_{s=1}^m(\gamma_{i_s}-H(Q))\big) \big)\nonumber \\
&&{\mbox{ \ \ \ \ }} \mid \gamma_{i_1},
\ldots , \gamma_{i_m} \in H^\ast(Q), \mu_2<\gamma_{i_1} < \cdots <
\gamma_{i_m} \leq \mu_1\bigg\}. \label{eqdejligdual}
\end{eqnarray}
\end{theorem}

\section{RGHWs of Hermitian codes}\label{secherm}\label{sec6}
In this section we apply the results of Section~\ref{sec5} 
 to the case of Hermitian codes \cite{tiersma,stichtenothhermitian}. Our main result is that~(\ref{eqthebsound}) is
 often tight. The Hermitian function
field over ${\mathbb{F}}_{q^2}$ ($q$ a prime power) is given by the
equation $x^{q+1}-y^q-y$ and it possesses exactly $q^3+1$
rational places which we denote 
$P_1, \ldots , P_{q^3},Q$ -- the last being the pole of $x$.
The Weierstrass semigroup of $Q$,
$H(Q)=\langle \rho(x)=q,\rho(y)=q+1\rangle$, has $g=q(q-1)/2$ gaps and conductor
$c=q(q-1)$. Let $D=P_1+\cdots +P_{q^3}$. In the following by a
Hermitian code we mean a code of the form
$C_{\mathcal{L}}(D,\mu Q)$. Clearly, this code is of length $n=q^3$. As
is well-known the dual of a Hermitian code is a Hermitian code. This
fact will be useful when in a later section we consider ramp
schemes based on Hermitian codes. We start our investigation with a
lemma that treats a slightly more general class of semigroups than the
semigroup $\langle q,q+1\rangle$ relevant to us.

\begin{lemma}\label{lem2}
Let $a$ be an integer, $a \geq 2$. Define $\Gamma=\langle a, a+1
\rangle$. For integers $m,\mu$ with $1 \leq m \leq  \mu \leq a+1$ it holds that  
\begin{eqnarray}
Z(\Gamma,\mu,m)=\sum_{s=0}^{m-2}(a-s)=a(m-1)-(m-2)(m-1)/2.\label{eq1}
\end{eqnarray}
\end{lemma}
\begin{proof}
Recall that a positive integer $\lambda$ is called a gap of $\Gamma$
if $\lambda \notin \Gamma$. All other non-negative integers are
called non-gaps. For the given semigroup $\Gamma$ the set of
non-negative integers consists of one non-gap followed by $a-1$ gaps,
then two non-gaps followed by $a-2$ gaps and so on up to $a-1$
non-gaps followed by $a-(a-1)=1$ gap. All the following numbers are
non-gaps. We denote the above maximal sequences of consecutive gaps
$G_1, \ldots , G_{a-1}$ with $\# G_v=a-v$, $v=1, \ldots ,a-1$ (such sequences are called
deserts in~\cite[Ex.\ 3]{MR3015351}).\\
First assume $1\leq m \leq \mu \leq a+1$. Let $-\mu \leq i_1 < \ldots
<i_{m-1} \leq -1$. We have 
\begin{eqnarray}
\#G_v \cap
\big(\cup_{s=1}^{m-1}(i_s+\Gamma)\big) \geq \min \{ \# G_v,m-1\}
\nonumber 
\end{eqnarray}
with equality when $i_{m-1}=-1, i_{m-2}=-2,$ $
\ldots, i_1=-(m-1)$. Summing up the contribution from all $G_v$ accounts for
$\sum_{s=1}^{m-2}(a-s)$. The term in~(\ref{eq1}) corresponding to
$s=0$, namely $a$, comes from considering the number
of negative integers in $\sum_{s=1}^{m-1}(i_s+\Gamma)$. Thus we have established~(\ref{eq1}).
\end{proof}
Recall from Theorem~\ref{the2} that we have three bounds on
the RGHW of
which~(\ref{eqthebsound}) is the weakest. Using Lemma~\ref{lem2}, for Hermitian codes
of codimension at most $q+1$, (\ref{eqthebsound}) translates
into the below closed formula expression~(\ref{eq3}). Surprisingly,
this expression is often equal to the true value of the RGHW.
\begin{theorem}\label{theresult}
Consider the Hermitian curve $x^{q+1}-y^q-y$ over
${\mathbb{F}}_{q^2}$. Let $P_1, \ldots , P_{n=q^3}$, and $Q$ be the rational
places and $D=P_1+\cdots +P_{n}$. Let $\mu_1, \mu_2$ be non-negative integers
with $1\leq \mu_1-\mu_2 \leq q+1$. For $1 \leq m \leq \dim
(C_{\mathcal{L}}(D,\mu_1Q))-\dim (C_{\mathcal{L}}(D,\mu_2Q))$ we have
\begin{eqnarray}
M_m(C_{\mathcal{L}}(D,\mu_1Q),C_{\mathcal{L}}(D,\mu_2Q)) &\geq&
n-\mu_1 + \sum_{s=0}^{m-2} (q-s) \label{eq3}\\
&=&n-\mu_1+q(m-1)-(m-2)(m-1)/2. \nonumber
\end{eqnarray}
If
\begin{eqnarray}
c -1 \leq \mu_2  {\mbox{ and }} \mu_1 < n-c.\label{eq4}
\end{eqnarray}
(recall that $c=q(q-1)$) then we have 
$\dim(C_{\mathcal{L}}(D,\mu_1Q))-\dim(C_{\mathcal{L}}(D,\mu_2Q))=\mu_1-\mu_2$ and
equality in~(\ref{eq3}).
\end{theorem}
\begin{proof}
Equation~(\ref{eq3}) is a consequence of the last part of 
Theorem~\ref{the2} and
the first part of Lemma~\ref{lem2}. The result concerning the dimensions
is well-known. That equality holds in~(\ref{eq3}) under 
condition~(\ref{eq4}) follows from Lemma~\ref{lemhejsa} below.
\end{proof}
\begin{lemma}\label{lemhejsa}
Let $\mu_1$ and $m$ be positive integers with $m \leq q+1$, $\mu_1 <
 n-c$ and $c-1<\mu_1-(m-1)$. Then there exist $m$ functions $f_0,
\ldots , f_{m-1}$ such that
\begin{itemize}
\item $f_i \in {\mathcal{L}}((\mu_1-i)Q)
  \backslash{\mathcal{L}}((\mu_1-(i+1))Q)$, $i=0, \ldots , m-1$.
\item The number of common zeros of $f_0, \ldots , f_{m-1}$ is exactly
  $\mu_1-\sum_{i=0}^{m-2}(q-i)$. 
\end{itemize}
\end{lemma}
\begin{proof}
As is well-known $\cup_{\mu=0}^\infty {\mathcal{L}}(\mu Q)$ is isomorphic to ${\mathbb{F}}_{q^2}[X,Y]/I$, where $I= \langle X^{q+1}-Y^q-Y
\rangle$. The isomorphism is given by $\varphi(x)=X+I$ and $\varphi(y)=Y+I$. We call $X^{q+1}-Y^q-Y={\mbox{N}}(X)-{\mbox{Tr}}(Y)$ the
Hermitian polynomial -- ${\mbox{N}}$ being the norm and ${\mbox{Tr}}$
the trace corresponding to the field extension
${\mathbb{F}}_{q^2}/{\mathbb{F}}_q$. In this description the rational places $P_1, \ldots ,
P_{q^3}$ correspond to the affine points of the Hermitian
polynomial. We remind the reader of the following few facts
which play a crucial role in the below induction proofs: 
\begin{itemize}
\item For any $\delta \in
{\mathbb{F}}_{q^2}$ we have ${\mbox{N}}(\delta), {\mbox{Tr}}(\delta)
\in {\mathbb{F}}_q$.
\item For every $\epsilon \in {\mathbb{F}}_q$ there
exists exactly $q$ different $\delta$ such that
${\mbox{Tr}}(\delta)=\epsilon$. 
\item There exist exactly $q+1$ different
$\delta$ such that ${\mbox{N}}(\delta)=1$.
\end{itemize} 
We start by fixing some notation. Let $\{\alpha_1, \ldots ,
\alpha_q\}$ be the elements in ${\mathbb{F}}_{q^2}$ that map to $1$
under Tr. Let $\{\beta_1, \ldots ,\beta_{q^2-(q+1)}\}$ be the elements
that do not map to $1$ under N and $\{\gamma_1, \ldots ,
\gamma_{q+1}\}$ the elements that do.\\
Write $\mu_1=iq+j(q+1)$ with $0 \leq j <q$. First assume $1 \leq m \leq
j+1$ and that $i < q^2-q$. By induction on $m$ (in this interval) one
can show that the set $\{F_0, F_1, \ldots , F_{m-1}\}$ where
\begin{eqnarray}
F_0&=&\big(\prod_{s=1}^i(X-\beta_s)\big)\big(\prod_{s=1}^j(Y-\alpha_s)\big),\label{hejeq1} \\
F_1&=&\big(\prod_{s=1}^i(X-\beta_s)\big)(X-\gamma_1)\big(\prod_{s=1}^{j-1}(Y-\alpha_s)\big),
  \ldots , \label{hejeq2} \\
F_{m-1}&=&\big(\prod_{s=1}^i(X-\beta_s)\big)\big(\prod_{s=1}^{m-1}(X-\gamma_s)\big)\big(\prod_{s=1}^{j-m+1}(Y-\alpha_s)\big), \label{hejeq3}
\end{eqnarray}
has exactly $iq+j(q+1)-\sum_{s=0}^{m-2}(q-s)$ zeros in common with
the Hermitian polynomial $X^{q+1}-Y^q-Y$ (we leave the technical
details for the reader). \\
Finally, assume $j+1 \leq m \leq j+q$. By induction on $m$ (in this
interval) one can show that the set $\{F_0,F_1, \ldots , F_{m-1}\}$
where  
\begin{eqnarray}
F_0&=&\big(\prod_{s=1}^{i-q+j}(X-\beta_s)\big)\big(\prod_{s=1}^{q-j}(X-\gamma_s)\big)\big(\prod_{s=1}^j(Y-\alpha_s)\big),\label{snabel1}
  \\
F_1&=&\big(\prod_{s=1}^{i-q+j}(X-\beta_s)\big)\big(\prod_{s=1}^{q-j+1}(X-\gamma_s)\big)\big(\prod_{s=1}^{j-1}(Y-\alpha_s)\big),\ldots
  ,\label{snabel2} \\
F_j&=&\big(\prod_{s=1}^{i-q+j}(X-\beta_s)\big)\big(\prod_{s=1}^{q}(X-\gamma_s)\big)
  ,\label{snabel3} \\
F_{j+1}&=&\big(\prod_{s=1}^{i-q+j}(X-\beta_s)\big)\big(\prod_{s=1}^{q-1}(Y-\alpha_s)\big)
  ,\nonumber \\
F_{j+2}&=&\big(\prod_{s=1}^{i-q+j}(X-\beta_s)\big)\big(\prod_{s=1}^{q-2}(Y-\alpha_s)\big)(X-\gamma_1),
  \ldots , \nonumber  \\
F_{m-1}&=&\big(\prod_{s=1}^{i-q+j}(X-\beta_s)\big)\big(\prod_{s=1}^{q-m+j+1}(Y-\alpha_s)\big)\big(\prod_{s=1}^{m-j-2}(X-\gamma_s)\big), \nonumber 
\end{eqnarray}
has exactly $iq+j(q+1)-\sum_{s=0}^m(q-s)$ zeros in common with
the Hermitian polynomial $X^{q+1}-Y^q-Y$ (again we leave the technical
details for the reader). For simplicity we covered
the case $m=j+1$ and $i < q^2-q$ in both induction proofs. Observe that the basis step
$m=j+1$ of the last induction proof corresponds to the terms
in~(\ref{snabel1}), (\ref{snabel2}), (\ref{snabel3}) which are
different from (\ref{hejeq1}), (\ref{hejeq2}), (\ref{hejeq3}) with $m=j+1$.
\end{proof}

For $1 \leq m \leq \mu_1-\mu_2\leq q+1$ but with $\mu_1$ and $\mu_2$ not satisfying the
condition in~(\ref{eq4}) we can often derive much better estimates
than~(\ref{eq3}).\\
For $\mu_2<c-1$ it may happen that not all of the numbers $\mu_1,
\mu_1-1, \ldots , \mu_1-(m-1)$ belong to $H(Q)$, and so the worst
case in the proof of Theorem~\ref{the2} may not be realized. Hence, we should rather apply~(\ref{eqihopla2lig}) or (\ref{eqdejlig})
(which in this situation are equivalent).\\
For $n-c \leq \mu_1$ it may happen that $H^\ast(Q) \backslash (\mu_1+H(Q))$ is
strictly smaller than $H(Q)\backslash(\mu_1+H(Q))$ (this will happen
if $\mu_1=iq+j(q+1)$, with $q^2-q\leq i < q^2$ and $0 < j < q$). In
such a case
$\# \big( H^\ast (Q) \cap (\mu_1+H(Q))\big)$ will be strictly larger
than $n-\mu_1$. Moreover, all the numbers $\mu_1, \mu_1-1, \ldots , \mu_1-(m-1)$ need
not belong to $H^\ast (Q)$ (this may happen if $\mu_1 \geq n$) and again the worst case considered in the
proof of Theorem~\ref{the2} may not be realizable. In this situation 
we should rather apply~(\ref{eqdejlig}).\\

We illustrate our observations with three examples. The first two are
concerned with $\mu_2<c-1$ and the last with $n-c \leq \mu_1$. 

\begin{example}\label{exx1}
In this example we consider codes over
${\mathbb{F}}_{q^2}={\mathbb{F}}_{16}$. Hence, $q=4$, $H(Q)=\langle
4,5\rangle$ and $n=64$. The first numbers of $H^\ast(Q)$ (and $H(Q)$) are
$0,4,5,8,9,10,12$. Hence, $dim C_{\mathcal{L}}(D,8Q)=4$, $dim
C_{\mathcal{L}}(D,12Q)=7$. Theorem~\ref{theresult} tells us that 
$M_m(C_{\mathcal{L}}(D,12Q),C_{\mathcal{L}}(D,8Q))$ is at
least $52$, $56$ and $59$, for $m$ equal to $1$, $2$ and $3$,
respectively. Using~(\ref{eqihopla2lig}) we now show that for $m=2$
and $m=3$ the true values are at least $58$ and $60$, respectively. We
first concentrate on $m=2$. Using the notation from
Proposition~\ref{the1} we must investigate all $\gamma_{i_1},
\gamma_{i_2} \in \{9,10,12\}$ with $\gamma_{i_1}<\gamma_{i_2}$, 
We have three different choices of
$(\gamma_{i_1},\gamma_{i_2})$ to consider, namely $(10,12)$, $(9,12)$ and
$(9,10)$. We first observe that
\begin{eqnarray}
12+H(Q)&=&\{12,16,17,20,21,22,24,\ldots\} \nonumber \\
10+H(Q)&=&\{10,14,15,18,19,20,22,23,24,\ldots \} \nonumber \\
9+H(Q)&=&\{9,13,14,17,18,19,21,22,23,24, \ldots \}. \nonumber
\end{eqnarray}
Note that if $\alpha \in H(Q) \backslash (\lambda+H(Q))$ for $\lambda
\in \{9,10,12\}$ then also $\alpha \in H^\ast (Q)$.\\

\noindent {\underline{$(\gamma_{i_1},\gamma_{i_2})=(10,12)$:}} 
We have 
\begin{eqnarray}
&\# (H^\ast(Q) \cap (12+H(Q))=n-12=52, \label{eqinox1}\\
&\# ((10+H(Q))\backslash (12+H(Q))=6.\nonumber
\end{eqnarray}
Hence, we get the value $52+6=58$.\\

\noindent {\underline{$(\gamma_{i_1},\gamma_{i_2})=(9,12)$:}} 
Combining (\ref{eqinox1}) with
$$\#((9+H(Q))\backslash(12+H(Q))=6$$
again give us the value $52+6=58$.\\

\noindent {\underline{$(\gamma_{i_1},\gamma_{i_2})=(9,10)$:}} We have
\begin{eqnarray}
&\# (H^\ast(Q) \cap (10+H(Q))=n-10=54, \nonumber\\
&\# ((9+H(Q))\backslash (10+H(Q))=4\nonumber
\end{eqnarray}
 producing the value $54+4=58$.\\

The minimum of the above three values is $58$ which is then our
estimate on $M_2(C_{\mathcal{L}}(D,12Q),C_{\mathcal{L}}(D,8Q))$.\\

Finally consider $m=3$. There is only one choice of
$(\gamma_{i_1},\gamma_{i_2},\gamma_{i_3})$ namely $(9,10,12)$. By inspection
there are exactly $8$ numbers that are in either $9+H(Q)$ or $10+H(Q)$
but not in $12+H(Q)$. Hence, our estimate on
$M_3(C_{\mathcal{L}}(D,12Q),C_{\mathcal{L}}(D,8Q))$ becomes $n-12+8=60$.
\end{example}

\begin{example}\label{exx2}
This is a continuation of Example~\ref{exx1}. The dimension of
$C_{\mathcal{L}}(D,10Q)$ and $C_{\mathcal{L}}(D,5Q)$ are $6$ and $3$,
respectively. Theorem~\ref{theresult} tells us that\\   
$M_m(C_{\mathcal{L}}(D,10Q), C_{\mathcal{L}}(D,5Q))$ 
is at least $n-10=54$, $n-10+4=58$ and $n-10+4+3=61$,
for $m$ equal to $1$, $2$ and $3$, respectively. The possible values of $\gamma_{i_s}$ to consider are $8,
9, 10$, which constitute a sequence without gaps. Hence, according to
our discussion prior to Example~\ref{exx1} in this case we cannot improve upon
Theorem~\ref{the2}.
\end{example}

\begin{example}\label{exx3}
This is a continuation of Example~\ref{exx1} and \ref{exx2}. The last
numbers of $H^\ast (Q)$ are $\{65,66,67,69,70,71,74,75,79\}$.
Hence, $\dim (C_{\mathcal{L}}(D,69Q))=64-5=59$ and $\dim
(C_{\mathcal{L}}(D,65Q))=64-8=56$. Theorem~\ref{theresult} gives no information
on the first two RGHWs and only tells
us that the third relative weight is larger than or equal to $2$. This,
however, is useless information as any space $D$ of
dimension $3$ has a support of size at least $3$. As we will now
demonstrate (\ref{eqdejlig}) guarantees that  
the three RGHWs are at least $3$, $6$,
and $8$, respectively. We first observe that 
\begin{eqnarray}
H^\ast (Q) \cap (69+H(Q))&=&\{69,74,79\} \nonumber \\
H^\ast (Q) \cap (67+H(Q))&=&\{67,71,75,79\} \nonumber \\
H^\ast (Q) \cap (66+H(Q))&=&\{66,70,71,74,75,79\}. \nonumber 
\end{eqnarray}
The smallest set is of size $3$ and we get 
$M_1(C_{\mathcal{L}}(D,69Q), C_{\mathcal{L}}(D,65Q))=3$. \\
The smallest union of two sets is the union of the first
two. This union is of size $6$ giving us 
$M_2(C_{\mathcal{L}}(D,69Q), C_{\mathcal{L}}(D,65Q))\geq 6$.\\
The union of all three sets is of size $8$. Hence,
$M_3(C_{\mathcal{L}}(D,69Q), C_{\mathcal{L}}(D,65Q))\geq 8$.
\end{example}

\subsection{A comparison between RGHW and GHW}

In \cite{MR1677034} and \cite{MR1737936}, respectively, 
Munuera \& Ramirez and Barbero \& Munuera  
determined the GHWs of any Hermitian
code. To state all their results is too extensive. However,
already from their master theorem~\cite[Prop.\ 12]{MR1677034},
\cite[Prop.\ 2.3]{MR1737936}, one can deduce that the RGHWs are often
much larger than the corresponding GHWs. 
\begin{definition}
Let ${\mbox{ev}} : \cup_{\mu=0}^\infty C_{\mathcal{L}}(D,\mu Q) \rightarrow
{\mathbb{F}}_q^n$ be the map ${\mbox{ev}}(f)=(f(P_1),$ $\ldots ,$
$f(P_n))$. The abundance $\alpha(\mu)$ is the dimension of $\ker {\mbox{ev}}$ when ${\mbox{ev}}$ is
restricted to $C_{\mathcal{L}}(D,\mu Q)$.
\end{definition}
\noindent The following is the master theorem
from~\cite{MR1677034,MR1737936}. Here, and
throughout the rest of this section, we use the notation $H(Q)=\{\rho_1,
\rho_2, \ldots \}$ with $\rho_i < \rho_j$ for $i<j$. 
\begin{theorem}\label{promunuera}
For $m=1, \ldots , \dim (C_{\mathcal{L}}(D,\mu Q))$ 
\begin{equation}
d_m(C_{\mathcal{L}}(D,\mu Q)) \geq n-\mu+\rho_m+\alpha(\mu). \label{eqpj1}
\end{equation}
Equality holds under the following conditions:
\begin{enumerate}
\item $\mu \in H^\ast(Q)$
\item $n-\mu+\rho_{m+\alpha(\mu)}\in H(Q)$, in which case we write
$n-\mu+\rho_{m+\alpha(\mu)}=iq+j(q+1)$, where $i, j$ are non-negative
  integers with $j < q$.
\item $i\leq q^2-q-1$ or $j=0$.
\end{enumerate}
\end{theorem}
Observe that Theorem~\ref{promunuera} and
Theorem~\ref{theresult}, respectively, produce similar estimates for the
minimum distance and the relative minimum distance. Similarly for the second GHW and the second
RGHW. From the last part of
Theorem~\ref{theresult} we conclude that for $m=1,2$, whenever  
$m \leq \mu_1-\mu_2 \leq q+1$, 
$c -1 \leq \mu_2$  and  $\mu_1 < n-c$ holds, then
$M_m(C_{\mathcal{L}}(D,\mu_1Q),C_{\mathcal{L}}(D,\mu_2Q))=d_m(C_{\mathcal{L}}(D,\mu_1Q))$
(recall that $c$ is the conductor). As shall be demonstrated in the
following, for higher values of $m$, $M_m(C_{\mathcal{L}}(D,\mu_1Q),C_{\mathcal{L}}(D,\mu_2Q))$ is often much
larger than $d_m(C_{\mathcal{L}}(D,\mu_1Q))$. 
\begin{proposition}\label{proorig}
For  
 $q>2$, $1\leq m \leq q+1$ and $2q^2-q\leq \mu \leq n-c$ we have
$d_m(C_{\mathcal{L}}(D,\mu Q))=n-\mu+\rho_m$.
\end{proposition}
\begin{proof}
It is well-known~\cite{tiersma} that for $\mu \leq q^3-1$ we have
$\alpha(\mu)=0$. 
Therefore~(\ref{eqpj1}) simplifies to $d_m(C_{\mathcal{L}}(D,\mu Q)) \geq
n-\mu+\rho_m$ under the conditions of the proposition. To prove the proposition it suffices to demonstrate the
conditions 1, 2, and 3 of Theorem~\ref{promunuera}. As is well-known
$\mu \in H^\ast (Q)$ when $c \leq \mu <n$. However, $c < 2q^2-q$ and therefore condition
1 follows. To see that condition 2 is satisfied note that by
assumption $c \leq n-\mu$ and so $n-\mu+\rho_{m+\alpha (\mu)} \geq c$. To demonstrate condition
3 it suffices to show 
\begin{equation}
n-\mu+\rho_m\leq q^3-q^2.\label{eqnyong}
\end{equation}
Observe that $\rho_m \leq q(q-1)$ which holds because of the
assumption that $m \leq q+1$ and $q >2 $ and because the number of
gaps in $H(Q)$ equals $q(q-1)/2$. As a consequence the assumption
$2q^2-q \leq \mu$ implies $q^2+\rho_m\leq \mu$ from which we derive~(\ref{eqnyong}).
\end{proof}

\begin{proposition}\label{propdifference}
Consider the field ${\mathbb{F}}_{q^2}$, with $q>2$. Let $3 \leq
\tilde{\mu} \leq
q+1$ be fixed. For $m=3, \ldots , \tilde{\mu}$ there are at least $q^3-3q^2+1$
different codes
$C_{\mathcal{L}}(D,\mu Q)$ for which
$d_m(C_{\mathcal{L}}(D,\mu Q))=n-\mu+\rho_m$ and
simultaneously
$M_m(C_{\mathcal{L}}(D,\mu Q),C_{\mathcal{L}}(D,(\mu-\tilde{\mu})Q))=n-\mu+\sum_{i=0}^{m-2}(q-i)$
hold. For these
codes we have 
\begin{eqnarray}
M_m(C_{\mathcal{L}}(D,\mu Q),C_{\mathcal{L}}(D,(\mu-\tilde{\mu})Q))
-d_m(C_{\mathcal{L}}(D,\mu Q))\nonumber \\
=
\big(\sum_{s=0}^{m-2}(q-s)\big)-\rho_m > 0.\label{eqdetvirker}
\end{eqnarray}
\end{proposition}
\begin{proof}
Follows from Theorem~\ref{promunuera}, Theorem~\ref{theresult} and a
study of $H(Q)$.
\end{proof}
Note that if for fixed $\tilde{\mu}$ we divide the number of different
codes $C_{\mathcal{L}}(D,\mu Q)$ for
which (\ref{eqdetvirker}) holds by the number of 
different codes, which is $q^3$, then we get the ratio $R(q)\geq (q^3-3q^2+1)/q^3 \geq
1-3/q$. This ratio approaches $1$ as $q$ approaches infinity. For $q=4, 5, 7, 8, 9, 16$, and $32$, respectively, $R(q)$ is at least $0.25$, $0.4$, $0.57$, $0.62$, $0.66$, $0.81$, and
$0.9$, respectively. In Table~\ref{tabhusto} for different
values of $m$ and $q$ we list the difference between the parameters as
expressed in~(\ref{eqdetvirker}).
\begin{table}
$$
\begin{tabular}{lrrrrrrrr}
m&3&4&5&6&7&8&9&10\\
{\mbox{Diff}}(m,4) &2&1&1\\
{\mbox{Diff}}(m,5)&3&2&3&3\\
{\mbox{Diff}}(m,7)&5&4&7&9&6&6\\
{\mbox{Diff}}(m,8)&6&5&9&12&9&10&10\\
{\mbox{Diff}}(m,16)&14&13&25&36&33&42&50&57\\
\ \\
m&11&12&13&14&15&16&17\\
{\mbox{Diff}}(m,16)&51&56&60&63&65&55&55
\end{tabular}
$$
\caption{Diff$(m,q)$ is the value of~(\ref{eqdetvirker}).}
\label{tabhusto}
\end{table}

\section{Ramp schemes based on Hermitian codes}\label{sec7}
In this section we consider ramp secret sharing schemes $D_1/D_2$
where $D_1=C_2^\perp$, $D_2=C_1^\perp$, and $C_2\subsetneq C_1$ are
Hermitian codes 
over ${\mathbb{F}}_{q^2}$, with $\dim (C_1)- \dim (C_2)=\tilde{\mu}$. Recall from Theorem~\ref{theconnection} in
Section~\ref{sec2} that $t_m+1=M_m(C_1,C_2)$, $m=1, \ldots , \tilde{\mu}$
is the size of the smallest group that can reveal $m$ $q^2$-bits of
information. Also recall that $r_m=n-M_{\tilde{\mu}-m+1}(D_1,D_2)+1$
is the smallest number such that any group of this size can reveal $m$
$q^2$-bits of information. From Section~\ref{secherm} we know how to
determine/estimate $M_m(C_1,C_2)$. Now~\cite[Th.\ 1]{tiersma} tells us
that for $\mu \in H^\ast (Q)$ we have $C_{\mathcal{L}}(D,\mu
Q)^\perp=C_{\mathcal{L}}(D,(n+c-2-\mu )Q)$.
 To establish information on $r_m$ we therefore need  
not apply Theorem~\ref{the2dual} (the theorem for duals of one-point algebraic
geometric codes), but can instead use the already established
information on the RGHW of $C_2\subseteq C_1$. From 
Theorem~\ref{theresult} we get the following result:
\begin{theorem}
Let $\mu ,\tilde{\mu}$ be positive integers satisfying 
\begin{eqnarray}
\tilde{\mu }\leq q+1, \, \, c-1+\tilde{\mu }\leq \mu \leq n-1. \label{eqnoget4}
\end{eqnarray}
Consider the ramp secret sharing scheme $D_1/D_2=C_2^\perp/C_1^\perp$ where $C_1=C_{\mathcal{L}}(D,\mu Q)$ and
$C_2=C_{\mathcal{L}}(D,(\mu -\tilde{\mu})Q)$. The codimension (and
thereby the length of the secret) equals $\tilde{\mu}$. Furthermore
for $m=1, \ldots , \tilde{\mu}$ it holds that
\begin{eqnarray}
t_m &\geq& n-\mu +\sum_{s=0}^{m-2}(q-s) -1, \label{eqhejsa1}\\
r_m &\leq&
n-\mu +c+\tilde{\mu}-1-\sum_{s=0}^{\tilde{\mu}-m-1}(q-s). \label{eqhejsa2}
\end{eqnarray}
Equality holds simultaneously in (\ref{eqhejsa1}) and (\ref{eqhejsa2})
when the second condition in~(\ref{eqnoget4}) is replaced with 
\begin{equation}
2c-2+\tilde{\mu } < \mu < n-c.\label{eqnybetingelse}
\end{equation}
\end{theorem}

\begin{example}\label{exknox}
In this example we consider schemes over ${\mathbb{F}}_{64}$. That is,
$q=8$ and the number of participants is $n=512$. The
assumption (\ref{eqnoget4}) for (\ref{eqhejsa1})
and (\ref{eqhejsa2}) to hold is $\tilde{\mu} \leq 9$, $55+\tilde{\mu} \leq \mu \leq 511$, the
latter corresponding to $1 \leq n-\mu \leq
457-\tilde{\mu}$. By~(\ref{eqnybetingelse}) equality
holds simultaneously in~(\ref{eqhejsa1})
and (\ref{eqhejsa2}) when $56 <n-\mu < 402-\tilde{\mu}$ holds. In
Table~\ref{tabtop} we list for $\tilde{\mu }=q+1$ the values of
$G_1(m,q)=\sum_{s=0}^{m-2}(q-s)$ (which is our lower bound on
$(t_m+1)-(n-\mu )$) and $G_2(m,\tilde{\mu},q)=c+\tilde{\mu}-1-\sum_{s=0}^{\tilde{\mu}-m-1}(q-s)$ (which is our upper bound on
$r_m-(n-\mu )$). Note that
$G_1(m,q)=Z(H(Q),\mu,m)$ (Lemma~\ref{lem2}).
\begin{table}
\begin{center}
\begin{tabular}{crrrrrrrrr}
$m$&1&2&3&4&5&6&7&8&9\\
$G_1(m,8)$&0&8&15&21&26&30&33&35&36\\
$G_2(m,9,8)$&28&29&31&34&38&43&49&56&64
\end{tabular}
\end{center}
\caption{Parameters of the ramp schemes in Example~\ref{exknox}.}
\label{tabtop}
\end{table}
For the considered choice of $\tilde{\mu }$ the secret is of size equal to $9$ $q^2$-bits. One can get much
information from Table~\ref{tabtop}. 
 Assume for instance
$n-\mu=130$. Then the smallest group that can derive some information is of
 size $130+0=130$, hence $t_1=129$. 
The smallest group size for which any group can
 derive some information is $r_1=130+28=158$. Groups of size $158$ on the other hand can
 never obtain more than $5$ $q^2$-bits of information as $G_1(5,8)\leq
 158-130
 < G_1(6,8)$. Some group of size
 $t_3+1 =130+15=145$ can derive at least $3$ $q^2$-bits of
 information,
however, $r_3=130+31=161$ is the smallest group size guaranteed to reveal $3$ $q^2$-bits
 of information. Any group of size $r_9=130+64=194$ can reveal the entire
 secret. Some group of size $t_9+1=130+36=166$ can reveal the entire secret whereas
 other groups of size $166$ can reveal no more than $4$ $q^2$-bits of
 information.
\end{example}
\begin{example}\label{exknoxy}
In this example we consider schemes over ${\mathbb{F}}_{256}$ . That
is, $q=16$ and the number of participants is
$n=4096$.  Assumption~(\ref{eqnoget4}) is $1 \leq n-\mu <
3857-\tilde{\mu }$ and
by~(\ref{eqnybetingelse}) equality holds in (\ref{eqhejsa1}) and (\ref{eqhejsa2}) simultaneously
if 
\begin{eqnarray}
240< n-\mu < 3618-\tilde{\mu}\label{eqasspt}
\end{eqnarray}
In Table~\ref{tabtop2} we list values of $G_1(m,16)$ and
$G_2(m,16,16)$ where the functions $G_1$ and $G_2$ are as in Example~\ref{exknox}.
\begin{table}
\begin{center}
\begin{tabular}{crrrrrrrr}
$m$&1&2&3&4&5&6&7&8\\
$G_1(m,16)$&0&16&31&45&58&70&81&91\\
$G_2(m,16,16)$&120&122&125&129&134&140&147&155\\
\ \\
$m$&9&10&11&12&13&14&15&16\\
$G_1(m,16)$&100&108&115&121&126&130&133&135\\
$G_2(m,16,16)$&164&174&185&197&210&224&239&255
\end{tabular}
\end{center}
\caption{Parameters of the ramp schemes in Example~\ref{exknoxy}.}
\label{tabtop2}
\end{table}
 Assuming (\ref{eqasspt}), then from the table we get the following
 information: Some groups of size $t_1+1=n-\mu$ may reveal $1$ $q^2$-bit
 of information whereas other groups of size $n-\mu+119$ cannot as $r_1=n-\mu+120$. Some group of
 size $t_{11}+1=n-\mu +115$ can reveal $11$ $q^2$-bits of information whereas some
 group of the same size can not reveal anything. Any group of size
 $n-\mu+135$ can for sure reveal $5$ $q^2$-bits of information and some
 group of the same size can reveal everything. Any group of size
 $r_{16}=n-\mu +255$ can reveal the entire secret.  
\end{example}

\begin{remark}
Assume that~(\ref{eqnoget4}) holds and let $m\leq \tilde{\mu }$. The
difference between the smallest size for which any group can reveal
$m$ $q^2$-bits of information and the smallest size for which some
group can reveal $m$ $q^2$-bits of information equals
$(n-M_{\tilde{\mu }+1-m}(C_2^\perp,C_1^\perp)+1)-M_m(C_1,C_2)$ which is
at most  
\begin{equation}
c+\tilde{\mu }-1-\sum_{s=0}^{\tilde{\mu }-m-1}(q-s)-\sum_{s=0}^{m-2}(q-s)\label{eqsnables}
\end{equation}
(with equality if $2c-2+\tilde{\mu } <\mu < n-c $). The maximum of~(\ref{eqsnables})
is attained at $m=1$ and $m=\tilde{\mu }$. The corresponding
``worst-case'' difference equals $c+\tilde{\mu }-1-\frac{\tilde{\mu
  }-1}{2}(2q-\tilde{\mu }+2)$. This
number is highest possible when $\tilde{\mu }=q$ and $\tilde{\mu }=q+1$, in which case it
equals the genus $g=(q^2-q)/2$.
\end{remark}
We conclude the section with an example in which we show how to
improve upon (\ref{eqhejsa1}) and (\ref{eqhejsa2}) when the condition (\ref{eqnybetingelse}) is not satisfied.

\begin{example}\label{exslutprut}
In this example we consider schemes over
${\mathbb{F}}_{16}$. That is, $q=4$ and the number of participants
is $n=64$. We consider secrets of length $3$. Hence, we require that 
$$\dim (C_1=C_{\mathcal{L}}(D,\mu_1Q)) - \dim
(C_2=C_{\mathcal{L}}(D,\mu_2Q))=3.$$
We have
$$H^\ast (Q)=\{0,4,5,8,9,10,12,13, \cdots , 62, 63, 65, 66,
67, 70, 71, 75\}$$
and therefore without loss of generality the possible choices of $(\mu_1,\mu_2)$ are
\begin{eqnarray}
&&\{(\mu_1^{(1)},\mu_2^{(1)}), \ldots ,
  (\mu_1^{(62)},\mu_2^{(62)})\}\nonumber \\
&=&\{(5,-1), (8,0),(9,4),(10,5),(12,8),(13,9),(14,10),(15,12), 
  \nonumber \\
&&\ldots ,(63,60),(65,61),(66,62),(67,63),(70,65),(71,66),(75,67)\},\nonumber
\end{eqnarray}
where for $(5,-1)$ we mean that $C_2$ equals $\{\vec{0}\}$.
In the following we calculate
\begin{eqnarray}
t_m&=&M_m(C_{\mathcal{L}}(D,\mu_1Q),C_{\mathcal{L}}(D,\mu_2Q))-1,\nonumber
\\
r_m&=&n-M_{\mu_2-\mu_1-m+1}(C_{\mathcal{L}}(D,(n+c-2-\mu_2)Q),C_{\mathcal{L}}(D,(n-c+2-\mu_1)Q))+1\nonumber
\\
&=&n-M_{\mu_2-\mu_1-m+1}(C_{\mathcal{L}}(D,(74-\mu_2)Q),C_{\mathcal{L}}(D,(74-\mu_1)Q))+1,\nonumber
\end{eqnarray}
 $m=1,2,3$, for all the above choices of $(\mu_1,\mu_2)$.\\

Recall from the discussion prior to Example~\ref{exx1} in
Section~\ref{secherm} that for some choices of $(\mu_1,\mu_2)$ 
we may achieve better estimates on the RGHW than~(\ref{eq3}). This is
done by
applying the method of  Example~\ref{exx1} and
Example~\ref{exx3} which corresponds to (\ref{eqihopla2lig}) and (\ref{eqdejlig}), respectively. Specifically for
$\mu_1=5,8,9,10,12,13$ we do not have 
\begin{eqnarray}
\{\mu_1, \mu_1-1, \mu_1-2\}
\subseteq H^\ast (Q)\label{eqdeterfredag}
\end{eqnarray}
 and to calculate $t_m$ we therefore apply the
method of Example~\ref{exx1}. By inspection, for $\mu_1=53$, $57$, $58$, $61$, $62$, $63$, $65$,
$66$, $67$, $70$, $71$, $75$ we have that  $H^\ast (Q) \backslash (\mu_1+H(Q))$ is
strictly smaller than $H(Q)\backslash(\mu_1+H(Q))$ and also for
some of these values, (\ref{eqdeterfredag}) does not hold
either. Hence, we apply  the method of Example~\ref{exx3}. In conclusion
the values of $\mu_1$ for which we can potentially obtain improved information on
$t_m$ are
\begin{eqnarray}
S_1&=&\{\mu_1^{(1)},\mu_1^{(2)},\ldots ,\mu_1^{(6)},\mu_1^{(46)},\mu_1^{(50)},\mu_1^{(51)},\mu_1^{(54)},\mu_1^{(55)},\ldots
, \mu_1^{(62)}\}\label{eqdyt} \\
&=&\{5, 8, 9, 10, 12, 13, 53, 57 , 58, 61, 62, 63, 65,
66, 67, 70, 71, 75\}.\nonumber
\end{eqnarray}
We next discuss $r_m$. Here,  a little care is needed in the analysis:
as an example for
$(\mu_1,\mu_2)=(\mu_1^{(4)},\mu_2^{(4)})=(10,5)$ we have
$C_2^\perp=C_{\mathcal{L}}(D,(74-\mu_2)Q)=C_{\mathcal{L}}(D,69Q)$, but this
code is the same as $C_{\mathcal{L}}(D,67Q)$ because $68$ and $69$ do
not belong to $H^\ast(Q)$. This phenomenon corresponds to the fact
that actually 
$C_{\mathcal{L}}(D,\mu_2^{(s)}Q)^\perp=C_{\mathcal{L}}(D,\mu_1^{(63-s)}Q)$,
$s=1, \ldots , 62$. Hence,
from~(\ref{eqdyt}) we see that 
the values of $\mu_1$ for which we can potentially derive improved
information regarding $r_m$ are
\begin{eqnarray}
S_2&=&\{\mu_1^{(63-1)},\ldots ,\mu_1^{(63-6)},\mu_1^{(63-46)},\mu_1^{(63-50)},\mu_1^{(63-51)},\mu_1^{(63-54)},\ldots
, \mu_1^{(63-62)}\}\nonumber \\
&=&\{5, 8, 9, 10, 12, 13, 14, 15, 16, 19, 20, 24, 65,
66, 67, 70, 71, 75\}.\nonumber
\end{eqnarray}
Applying a mixture of the method from Example~\ref{exx1} and
Example~\ref{exx3} plus~(\ref{eq3}) we derive for $\mu_1\in S_1 \cup
S_2$ the information given in Table~\ref{tabtab}.\\
\begin{table}
\begin{tabular}{c|cccccccc}
$\mu_1$&5&8&9&10&12&13&14&15\\
$[t_1 ,r_1]$&[58,62]&[55,61]&[54,60]&[53,59]&[51,58]&[50,57]&[49,56]&[48,56]\\
$[t_2,r_2 ]$&[62,63]&[59,62]&[58,61]&[57,60]&[57,60]&[54,59]&[53,58]&[52,58]\\
$[t_3 ,r_3 ]$&[63,64]&[62,63]&[61,63]&[60,62]&[59,62]&[58,62]&[56,61]&[55,61]\\
\\
$\mu_1$&16&19&20&24&53&57&58&61\\
$[t_1 ,r_1 ]$&[47,55]&[44,52]&[43,51]&[39,47]&[11,18]&[7,14]&[7,13]&[3,10]\\ 
$[t_2,r_2 ]$&[51,58]&[48,54]&[47,54]&[43,50]&[14,21]&[10,17]&[10,16]&[6,13]\\
$[t_3,r_3 ]$&[54,61]&[51,57]&[50,57]&[46,53]&[17,25]&[13,21]&[12,20]&[9,17]\\
\\
$\mu_1$&62&63&65&66&67&70&71&75\\
$[t_1,r_1]$&[3,9]&[3,8]&[2,6]&[2,5]&[2,4]&[1,3]&[1,2]&[0,1]\\
$[t_2 ,r_2]$&[6,12]&[6,11]&[5,10]&[4,7]&[4,7]&[3,6]&[2,5]&[1,2]\\
$[t_3,r_3]$&[8,16]&[8,15]&[7,14]&[6,13]&[5,11]&[4,10]&[3,9]&[2,6]
\end{tabular}
\caption{Lower bounds on $t_m$ and upper bounds on $r_m$ for the schemes in Example~\ref{exslutprut}.}
\label{tabtab}
\end{table}
For the remaining values of $\mu_1$, that is for  
\begin{eqnarray}
\mu_1&\in&\{5, 8, 9, 10, 12, \ldots , 63, 65, 66, 67, 70, 71, 75\}
\backslash (S_1 \cup S_2) \nonumber \\
&=&\{17, 18, 21, 22, 23, 25, 26, 27, \ldots , 51, 52, 54, 55, 56, 59,
60\}\nonumber 
\end{eqnarray}
we have $\mu_2=\mu_1-3$,
and the best bounds (sometimes tight) are obtained
  from~(\ref{eq3}). They are:
$[t_1\geq n-\mu_1-1,r_1\leq n-\mu_1+7]$, 
$[t_2\geq n-\mu_1+3,r_2\leq n-\mu_1+10]$ and 
$[t_3\geq n-\mu_1+6,r_3\leq n-\mu_1+14]$.
\end{example}

\section*{Acknowledgments}
The authors gratefully acknowledge the support from
the Danish National Research Foundation and the National Natural Science
Foundation of China (Grant No.\ 11061130539) for the Danish-Chinese
Center for Applications of Algebraic Geometry in Coding Theory and
Cryptography. Furthermore the authors are thankful for the support
from Japan Society for the
Promotion of Science (Grant Nos.\ 23246071 and
26289116), from The Danish Council for Independent Research (Grant
No.\ DFF--4002-00367),  
from the Spanish MINECO
(Grant No.\ MTM2012-36917-C03-03), from  
National Basic Research Program of China (Grant No.\ 2013CB338004), and
National Natural Science Foundation of China (Grant No.\ 61271222). The authors would like to thank Ignacio Cascudo, Hao
Chen, Ronald Cramer and Carlos Munuera for pleasant discussions. Also
the authors would like to thank the anonymous reviewers for
valuable comments that helped us improve the paper.


\end{document}